\documentclass[10pt,twoside,a4paper]{article}         

\usepackage[T1]{fontenc}
\usepackage[english]{babel}
\usepackage[pdftex]{graphicx}           
\usepackage{setspace}                   
\usepackage{indentfirst}               
\usepackage{courier}                    
\usepackage{type1cm}                    
\usepackage{listings}                  
\usepackage{titletoc}
\usepackage[fixlanguage]{babelbib}
\usepackage[font=small,format=plain,labelfont=bf,up,textfont=it,up]{caption}
\usepackage[usenames,svgnames,dvipsnames]{xcolor}
\usepackage[a4paper,top=2.54cm,bottom=2.0cm,left=2.0cm,right=2.54cm]{geometry}
\usepackage[pdftex,plainpages=false,pdfpagelabels,colorlinks=true,citecolor=DarkGreen,linkcolor=NavyBlue,urlcolor=DarkRed,filecolor=green,bookmarksopen=true]{hyperref} 
\usepackage[all]{hypcap}                
\fontsize{60}{62}\usefont{OT1}{cmr}{m}{n}{\selectfont}
\usepackage{dsfont}
\usepackage{amsthm}
\usepackage{amsfonts}
\usepackage{amsmath}
\usepackage{amssymb}
\usepackage{mathrsfs}
\usepackage{lmodern}
\usepackage{relsize}
\usepackage{url}
\theoremstyle{plain}
\usepackage{mathtools}
\usepackage{thmtools}
\usepackage{enumitem}
\usepackage{calrsfs}
\usepackage{epsfig}
\usepackage{amsthm}

\newtheorem{teo}{Theorem}[section]
\newtheorem{prop}{Proposition}[section]
\newtheorem{lemma}{Lemma}[section]

\newtheorem*{remark}{Remark}
\newtheorem{remark2}{Remark}
\newtheorem{definition}{Definition}[section]
\theoremstyle{definition}

\numberwithin{equation}{section}

 \newcommand{\pp}{\mathsf{p}}
\newcommand{\vv}{\mathsf{V}}

\title{Ergodicity versus non-ergodicity for Probabilistic Cellular Automata on rooted trees}
\author{Bruno Kimura\\
\footnotesize{Delft Institute of Applied Mathematics}\\
\footnotesize{\texttt{bkimura@tudelft.nl}}\\
\footnotesize{Technische Universiteit Delft, Van Mourilk Broekmanweg 6, 2628 XE Delft , The Netherlands}\\
[0.3cm] Wioletta Ruszel\\
\footnotesize{Delft Institute of Applied Mathematics}\\
\footnotesize{\texttt{W.M.Ruszel@tudelft.nl}}\\
\footnotesize{Technische Universiteit Delft, Van Mourilk Broekmanweg 6, 2628 XE Delft , The Netherlands}\\
[0.3cm] Cristian Spitoni\\
\footnotesize{Institute of Mathematics}\\
\footnotesize{\texttt{C.Spitoni@uu.nl}}\\
\footnotesize{Universiteit Utrecht, Budapestlaan 6, 3584 CD Utrecht , The Netherlands}\\}
\date{\today}

\begin{document}

\maketitle


\begin{abstract}
In this article we study a class of shift-invariant and positive rate probabilistic cellular automata (PCAs) on rooted d-regular trees $\mathbb{T}^d$.  In a first result we extend the results of \cite{pca} on trees, namely  we prove that to every stationary measure $\nu$ of the PCA  we can associate a space-time Gibbs measure $\mu_{\nu}$ on $\mathbb{Z} \times \mathbb{T}^d$. Under  certain assumptions on the dynamics  the converse is also true.

A second result concerns proving sufficient conditions for ergodicity and  non-ergodicity of our PCA on d-ary trees for $d\in \{ 1,2,3\}$ and characterizing the invariant Bernoulli product measures.
\end{abstract}


\section{Introduction}
Cellular Automata (CAs) are discrete-time dynamical systems on a 
spatially extended discrete space. They are well known for being easy to implement and
for exhibiting a rich and complex nonlinear behaviour as 
emphasized for instance in~\cite{Wolfram1984,Wolfram1984Nature,NonLinearBook,Kari}; furthermore, they
can give rise to multiple levels of organization~\cite{hoekstra}. 
Probabilistic Cellular Automata (PCAs) whose the updating rule is now considered to be stochastic, see \cite{sirakoulis2012cellular}, are a straightforward generalization of CAs and are employed as modeling tools in 
a wide range of applications, e.g. HIV infection \cite{pandey}, biological immune system \cite{tome},  weather forecast  \cite{dorre}, heart pacemaker tissue \cite{danuta}, and opinion forming  \cite{bagnoli}. Moreover, a natural context in which the PCAs main ideas are of interest 
is that of evolutionary games \cite{PGSFM,PG,PS}.

Strong relations exist as well  between PCAs and the general equilibrium statistical 
mechanics framework~\cite{Wolfram:RevModPhys.55.601,PCA:order:disorder,LMS, Derrida,GLD,  pca,GJH, PSS, tvs,Vas69}. A central question is the characterization of  the equilibrium behavior of a general  PCA dynamics. For instance, one primary interest is the study of its ergodic properties, e.g. the long-term behavior of the PCA and its 
dependence on the initial probability distribution. Regarding the ergodicity for PCAs on infinite lattices, see for instance  \cite{tvs} for details and references. Moreover,  conditions for ergodicity for general PCAs can be found in the
following papers: \cite{LMS, ferrari, louis2,maes, malyshev, Steiff}.
Furthermore, in case of a  translation-invariant PCA on $\mathbb{Z}^d$  with positive rates, it has been shown in \cite{pca}  that
the law of the trajectories, starting from any stationary distribution, is given by a Gibbs state for some space-time associated potential (in $\mathbb{Z}^{d+1}$). Moreover, it has also been proven that the converse is true:  all the translation-invariant Gibbs states for such potential correspond to statistical
space-time histories for the PCA. Therefore,  phase transition for the space-time potential is closely related to the PCA
ergodicity in the sense that non-uniqueness of translation invariant Gibbs states is equivalent to non-uniqueness of stationary measures for the PCA. The main ingredient for proving this result
is the use of the local variational principle for the entropy density of the Gibbs measure. However, as it has been proved in \cite{Bur}, the variational principle for Gibbs states fails for nearest neighbor finite state statistical mechanics systems on $3$-ary trees. 
Hence, a first result to this paper is to extend the results presented by \cite{pca} for a class of PCAs on infinite rooted trees. In particular, the PCAs considered in this paper have 
 positive rate shift-invariant local transition probabilities such that each local probabilistic rule  depends only on the spins of the children of the node. This class of PCAs has generally a Bernoulli product measure as an invariant measure, and they are the natural generalization on trees of the models considered in \cite{MM1}.

A second type of results in this paper is to give conditions for ergodicity in case of $d$-ary trees, with $d\in\{1,2,3\}$. Our positive rate PCAs satisfy indeed such conditions (i.e.  (\ref{lem:statsol}) and 
(\ref{bodybuilder2})) that, when iterating the dynamics from the Bernoulli product measure, the resulting space-time diagram defines non-trivial random fields with very weak dependences.
This fact allows us to give a detailed analysis of the ergodicity problem and, for two relevant examples of PCA dynamics, we are able to find the critical parameters.

The paper is organized as follows. In Section \ref{PCAG} we extend the results of \cite{pca} in case of infinite rooted  $d$-ary trees. We first define the PCA on a countably infinite set and in this
general framework we show how stationary measures for a PCA can be naturally associated to Gibbs measures (Theorem~\ref{stm}). In order to state the converse result,  
we first restrict ourselves to the case of infinite rooted trees and to PCAs with nondegenerate shift-invariant local transition probabilities that depends only on the spins of the children of the node.
For this class of PCAs, we state that all the time-invariant Gibbs states for the potential correspond to statistical
space-time histories for the PCA (Theorem~\ref{drift}). In Section \ref{ENE} we give results concerning conditions for the ergodicity of the PCA on $d$-ary trees. First we characterize Bernoulli product stationary measures via Lemma~\ref{lem:statsol}. In Theorem~\ref{Ergo_1} we show that for $d=1$ the PCA is always ergodic, and the same occurs for $d=2$ with the additional assumption of spin-flip symmetry of the local transition
probabilities. In Theorem~\ref{Th:Ergo_2} the case of $d=3$ is studied. We give two examples (Section~\ref{Ergo_2}, Section~\ref{Ergo_4}) where the critical parameters can be computed. Section \ref{proofs}  and the Appendices are devoted to the proofs of the main results.

\section{From PCAs to Gibbs measures and back}\label{PCAG}
\subsection{PCAs on countably infinite sets}
Let the single spin space be a nonempty finite set $S$ and let $\mathsf{V}$ denote a 
countably infinite set (for example, the $d$-dimensional cubic lattice $\mathbb{Z}^{d}$ or, more generally, the vertex set of a countably 
infinite graph). In the following we introduce a special class of discrete-time Markov chains on the state space $\Omega_{0} = S^{\mathsf{V}}$ 
whose main feature is the fact that
given the previous configuration, for the next one 
all spins are simultaneously updated accoding to independent local transition probabilities (\emph{parallel updating}), the so-called probabilistic cellular automata. 

We define the probabilistic cellular automaton as follows. 
\begin{definition} \label{defi1}
A PCA is a discrete-time Markov chain on $\Omega_0$ with the following properties.
At each site $i$ in $\mathsf{V}$
\begin{enumerate}
\item[(a)] corresponding to each configuration $x \in \Omega_{0}$ we associate a probability measure $p_{i}(\cdot|x)$ on $S$, and
\item[(b)] assume that for every spin $s$, the map
\[x \mapsto  p_{i}(s|x)\]
is a local function. So, there is a finite subset $U(i)$ of $\mathsf{V}$ such that the equality $p_{i}(s|x) = p_{i}(s|y)$ holds for every $s$ whenever 
$x$ and $y$ satisfy $x_{j}=y_{j}$ for each $j$ in $U({i})$.
\end{enumerate}
In this setting, we associate to each point $x$ in $\Omega_{0}$ the product measure 
\begin{equation}\label{pcadef}
P(dy|x) = \bigotimes\limits_{i \in \mathsf{V}} p_{i}(dy_{i}|x),	
\end{equation}
and introduce the probabilistic cellular automaton dynamics on our state space $\Omega_{0}$ by considering the Markov kernel 
$P$ given by the expression
\begin{equation}
P(x,B) = P(B|x)	
\end{equation}
where $B$ is a Borel set of $\Omega_{0}$.
\end{definition}
Now, we recall the definition of a stationary measure for the dynamics $P$.
\begin{definition}\label{stamrs}
A probability measure $\nu$ on $\Omega_0$ is called stationary for the dynamics $P$ defined above if 
\[
\int P(x, B) \nu(dx) = \nu (B)
\]
holds for every Borel set $B$ of $\Omega_{0}$.
\end{definition}
\subsection{From PCA to Gibbs measures...}
In this section we will show how stationary measures for a PCA can be naturally associated to Gibbs measures for a corresponding equilibrium statistical mechanical model.
Let us consider the set of sites given by the countably infinite set $\mathbb{Z} \times \mathsf{V}$,
the collection $\mathscr{S}$ consisting of all nonempty finite subsets of $\mathbb{Z} \times \mathsf{V}$.
We also consider the configuration space $\Omega = S^{\mathbb{Z} \times \mathsf{V}}$ together with its product $\sigma$-algebra $\mathscr{F}$. 
Given an arbitrary space-time spin configuration $\omega$ in $\Omega$, 
for each site $\mathsf{x}$ in $\mathbb{Z} \times \mathsf{V}$, say $\mathsf{x} = (n,i)$,
let $\omega_{n,i}$ denote the value $\omega_{\mathsf{x}}$ of the spin at this site, just for simplicity. 
Furthermore, for each integer $n$ and each configuration 
$\omega$, we define the configuration at time $n$ as the
element $\omega_{n}$ of $\Omega_{0}$ given by $\omega_{n} = (\omega_{n,i})_{i \in \mathsf{V}}$. 

Now, let us consider again the setting from the previous section. We will assume that the PCA dynamics is nondegenerate, that is, the local transition
probabilities have positive rates: $p_{i}(s|x)>0$ holds for all $i \in \mathsf{V}$, $s \in S$ and $x \in \Omega_{0}$. Furthermore, we also 
suppose that for each site $i$, the set
\begin{equation}\label{cccp}
\{j \in \mathsf{V}: i \in U(j)\}	
\end{equation}
is finite, which means that at each step in the dynamics of the PCA, each spin can have influence only on the future state of a finite number of spins. Given a stationary measure $\nu$ 
for $P$, it is possible to construct a probability measure $\mu_{\nu}$ on $(\Omega,\mathscr{F})$ uniquely determined by the identity 
\begin{equation}\label{muu}
\mu_{\nu}(\omega_{t} \in B_{0}, \omega_{t+1} \in B_{1}, \dots, \omega_{t+n} \in B_{n}) = \int_{B_{0}}\nu(dx_{0})\int_{B_{1}}P(x_{0},dx_{1})\cdots\int_{B_{n}}P(x_{n-1},dx_{n}),
\end{equation} 
where $t$ is an integer, $n$ a positive integer, and $B_{0}, B_{1}, \dots, B_{n}$ are Borel sets of $\Omega_{0}$.
In the following, given a site $\mathsf{x}$ in $\mathbb{Z}\times \mathsf{V}$, say $\mathsf{x} = (n,i)$, we will use $U(\mathsf{x})$ to denote the set
\begin{equation*}
U(\mathsf{x}) = \{(n-1,j): j \in U(i)\}.	
\end{equation*}
Observe that our assumption (\ref{cccp}) is equivalent to say that for each point $\mathsf{x}$, the set
\[\{\mathsf{y} \in \mathbb{Z}\times \mathsf{V}: \mathsf{x} \in U(\mathsf{y})\}\]
is finite. This remark is very useful for proving the next theorem, whose proof is given in {Appendix \ref{app1}}.

\begin{teo}\label{stm}
The space-time measure $\mu_{\nu}$ obtained from a stationary measure $\nu$ for the PCA is a Gibbs measure for the interaction $\Phi = (\Phi_{A})_{A \in \mathscr{S}}$, where each 
$\Phi_{A}: \Omega \rightarrow \mathbb{R}$ is given by 	
\begin{equation}\label{int}
\Phi_{A}(\omega) = 
\begin{cases}
-\log p_{i}(\omega_{\mathsf{x}}|\omega_{n-1}) &\text{if $A = \{\mathsf{x}\}\cup U(\mathsf{x})$ for some $\mathsf{x} = (n,i)$,}\\
0 & \text{otherwise.}
\end{cases}	
\end{equation}
\end{teo}

\subsection{PCA on infinite rooted trees}
We specify now the class of PCAs that will be considered in this paper. We introduce indeed probabilistic cellular automata on $d$-ary trees $\mathsf{V}=\mathbb{T}^d$ with root $\mathbf{o}$ and degree $\deg(x)=d+1$ for all vertices $x\neq \mathbf{o}$ and $\deg(\mathbf{o})=d$.
 Without loss of generality, 
the $d$-ary tree $\mathbb{T}^d$ can be regarded as the set 
\[\bigcup\limits_{n \geq 0}\{0,\dots,d-1\}^{n}\]
consisting of all finite sequences of integers from $0$ to $d-1$. Given finite sequences $i$ in $\{0,\dots,d-1\}^{n}$ and $j$ in $\{0,\dots,d-1\}^{m}$, 
say $i=(i_{k})_{k=0}^{n-1}$ and $j = (j_{k})_{k=0}^{m-1}$,
we naturally define their sum $i+j$ as the concatenation of these sequences, i.e., the sum is the element of $\{0,\dots,d-1\}^{m+n}$ given by
\begin{eqnarray*}
(i+j)_{k} =
\begin{cases}
i_{k}&\text{if}\;k \in \{0,\dots,n-1\},\\
j_{k-n}&\text{if}\;k \in \{n,\dots,m+n-1\}.
\end{cases}	
\end{eqnarray*}  	
Once defined the translation on $\mathbb{T}^{d}$, then we are allowed to associate to each site $i$ in $\mathbb{T}^{d}$ 
the shift map $\Theta_{i} : S^{\mathbb{T}^{d}} \rightarrow S^{\mathbb{T}^{d}}$ defined by
\begin{equation}
\Theta_{i}x = (x_{i+j})_{j \in \mathbb{T}^{d}}	
\end{equation}
at each point $x=(x_{j})_{j \in \mathbb{T}_{d}}$. Furthermore, for each $k \in \{0,\dots,d-1\}$, we denote by $e_{k}$ the sequence $e_{k}=(k)$ consisting only of the number $k$, therefore, 
the $e_{k}$'s are the neighbors of the root $\mathbf{o}$ of $\mathbb{T}^{d}$. 

\medskip

From now on, we consider the single spin space $S=\{-1,+1\}$, so, the state space $\Omega_0$ is described as $\Omega_0 = \{-1,+1\}^{\mathbb{T}^{d}}$. 
Following \cite{dobrushin, louisthesis}, we give the definitions of attractive dynamics and of repulsive dynamics. In order to do that we introduce the notation $x \leq y$ to indicate that $x$ and $y$ are elements of $\Omega_0$ that satisfy $x_i \leq  y_i$ for all $i \in \mathbb{T}^{d}$. 


\begin{definition}
\label{def_attr}
We call the dynamics $P$ {\it attractive} if for every positive integer $n$, for all configurations $x,y$ such that $x\leq y$ and each nondecreasing local function $f$, we have
\begin{equation}
\label{attr1}
P^n(x,f)\leq P^n(y,f).
\end{equation}
\end{definition}

\begin{definition}
We call the dynamics $P$ {\it repulsive} if for every positive integer $n$, for all configurations $x,y$ such that $x\leq y$ and each nondecreasing local function $f$, we have 
\begin{equation}
\label{repul1}
P^n(x,f)\ge P^n(y,f).
\end{equation}
\end{definition}

By \cite{dobrushin,louisthesis} it follows that the dynamics is  {\it attractive} if and only if for all configurations $x,y$ such that $x\leq y$ we have $p_{\mathbf{o}}(+1|x) \leq p_{\mathbf{o}}(+1|y)$; furthermore, it is {\it repulsive} if and only if
for all configurations $x,y$ such that $x\leq y$ we have $p_{\mathbf{o}}(+1|x) \geq p_{\mathbf{o}}(+1|y)$.

\medskip

The PCAs considered in this paper have nondegenerate shift-invariant local transition probabilities such that each probabilistic rule 
$p_{i}(\cdot|x)$ depends only on the spins of the children of $i$.  More precisely, we will state the following assumptions on the transition kernel.\\
\bigskip

\noindent
{\bf Assumptions:}
\begin{enumerate}
\item[(A1)] each $p_{\mathbf{o}}(\cdot|x)$ is a probability measure such that $p_{\mathbf{o}}(s|x)>0$ holds for all $s \in \{-1,+1\}$,
\item[(A2)] the map $x \mapsto p_{\mathbf{o}}(s|x)$ depends only on the values of
$x$ on $U(\mathbf{o})=\{e_{0},\dots,e_{d-1}\}$, and
\item[(A3)]   for each $i$ in $\mathbb{T}^{d} \backslash \{\mathbf{o}\}$, the local transition probability $p_{i}(\cdot|x)$ satisfies 
\begin{equation}
p_{i}(s|x) = p_{\mathbf{o}}(s|\Theta_{i}x).	
\end{equation}
\end{enumerate}
 Note that Assumption (A1) is the so-called nondegeneracy property, while Assumption (A3) is the invariance of the PCA dynamics under tree shifts.
We remark as well that, it follows from $(A2)$ and $(A3)$ that the map $x \mapsto p_{i}(s|x)$ depends only on the values assumed by the spins of $x$ on 
$U(i) = i + \{e_{0},\dots,e_{d-1}\}$. 
One of the crucial
features of this dynamics $P$ is that under Assumptions (A2) and (A3) the relation
\begin{equation}\label{prodoido}
P^{n}(x,\{y_{F} = \xi\}) = \prod\limits_{i \in F}P^{n}(\Theta_{i}x,\{y_{\mathbf{o}} = \xi_{i}\})	
\end{equation}
holds for every configuration $x$, finite volume configuration $(\xi_{i})_{i \in F}$ for some $F \subseteq \mathbb{T}^d$, and positive integer $n$.

\subsection{...and back}
According to Theorem \ref{stm}, every stationary measure for the PCA defined 
above can be associated to a Gibbs measure for the corresponding statistical mechanical model $\Phi$ defined by (\ref{int}). Next, we show that for the class of PCAs on trees we are dealing with, under suitable conditions, the converse is also valid.

\begin{teo}\label{drift}
Under the Assumption (A1)-(A3), let $\mu$ be a Gibbs measure for the interaction $\Phi$ defined by (\ref{int}), such that it is invariant under time translations, i.e., $\mu$ is a Gibbs measure that satisfies 
\[\mu(\omega_{m} \in B) = \mu(\omega_{m-1} \in B)\]
for each integer $m$ and each Borel subset $B$ of $\Omega_{0}$. Then, there is a stationary measure $\nu$ for the corresponding PCA
such that $\mu = \mu_{\nu}$.		
\end{teo}
Therefore, thanks to Theorem~\ref{drift} the study of the ergodicity of the PCA can be closely related to the study the uniqueness of the Gibbs measure on space-time associated to it.
\begin{remark}
In Appendix \ref{app1}, we give a more general proof for Theorem~\ref{drift}. It actually holds for any PCA on $\Omega_0 = S^{\mathsf{V}}$, where $S$ is a nonempty finite set and $\mathsf{V}$ is a (locally finite) infinite rooted tree, satisfying $(A1)$ and
\begin{enumerate}
\item[(A2')] Let $d:  \mathsf{V}  \times \mathsf{V} \rightarrow \mathbb{R}$ be the distance function that assigns to each pair $(i,j)$ of vertices the length of the unique path connecting them. 
Corresponding to each point $i$ that belongs to $\mathsf{V}$ the set $U(i)$ is a finite set such that 
\begin{equation}\label{pcatree}
U(i) \subseteq \{j \in \ \mathsf{V} : d(\mathbf{o},i) < d(\mathbf{o},j)\}.
\end{equation} 
\end{enumerate}
\end{remark}

\section{Conditions for ergodicity for PCA on trees}\label{ENE}
\label{s:ergo}
In this section we will present some results regarding sufficient conditions for ergodicity for the class of PCAs described previously. Note that equation (\ref{prodoido}) implies that the probability distributions of the spins at time $n$ are independent, so, this suggests that the typical stationary measures we have to look for are product measures. This remark leads us to state a lemma regarding the 
 characterization of stationary Bernoulli product measures, whose proof is given in Appendix \ref{app2}.
\begin{lemma}\label{lem:statsol}
A Bernoulli product measure $\nu = \mathsf{Bern(p)}^{\otimes \mathbb{T}^{d}}$ with parameter  $\mathsf{p} \in [0,1]$, 
is a stationary measure for $P$ if and only if 
\begin{equation}\label{vemmonstro}
\int p_{\mathbf{o}}(+1|x) \nu(dx) = \mathsf{p}
\end{equation}
i.e. if and only if
\begin{equation}\label{bodybuilder}
\sum\limits_{l=0}^{d}(-1)^{l}\left[\sum\limits_{\substack{I \subseteq \{0,\dots,d-1\} \\ |I| = l}}\left(\sum\limits_{\substack{\xi \in \{-1,+1\}^{d} \\ \xi_{k}= -1\; \textnormal{for all}\; k \in I}}(-1)^{\#\{m:\xi_{m}=-1\}}p_{\mathbf{o}}(+1|\xi)\right)\right]\mathsf{p}^{d-l} = \mathsf{p}.	
\end{equation}
Moreover, the probability to find the spin $+1$ at the root of $\mathbb{T}^{d}$ after $n+1$ steps of this dynamics starting from the configuration $x$ can be written as
\begin{eqnarray}\label{bodybuilder2}
&&P^{n+1}(x,\{y_{\mathbf{o}} = +1\})\\
\nonumber\\
&&=\sum\limits_{l=0}^{d}(-1)^{l}\left[\sum\limits_{\substack{I \subseteq \{0,\dots,d-1\} \\ |I| = l}}\left(\sum\limits_{\substack{\xi \in \{-1,+1\}^{d} \\ \xi_{k}= -1\; \textnormal{for all}\; k \in I}}(-1)^{\#\{m:\xi_{m}=-1\}}p_{\mathbf{o}}(+1|\xi)\right)\prod_{k \in \{0,\dots,d-1\}\backslash I}P^{n}(\Theta_{e_{k}}x,\{y_{\mathbf{o}} = +1\})\right] \nonumber.	
\end{eqnarray}
\end{lemma}

\subsection{Ergodicity and examples}
From now on, we will abbreviate $+1$ by $+$ (resp. $-1$ by $-$). 
In the first theorem we prove ergodicity results for the line and the binary trees,  while in the second theorem we prove ergodicity and non-ergodicity results for the 3-ary trees.  
\begin{teo}\label{Ergo_1}
Let us consider a PCA with transition probabilities satisfying (A1)-(A3). Then, we have the following results. 
\begin{enumerate}
\item[(a)] If $d=1$, then the PCA dynamics is ergodic.  The unique stationary measure is a Bernoulli product measure with parameter
\begin{equation}\label{1dstat}
\mathsf{p}=\frac{p_{\mathbf{o}}(+|-)}{p_{\mathbf{o}}(-|+) + p_{\mathbf{o}}(+|-)}.    
\end{equation}

\item[(b)] Let $d=2$ and the transition probabilities being symmetric under spin-flip, i.e., the equality $p_{\mathbf{o}}(s|x)=p_{\mathbf{o}}(-s|-x)$ holds for every spin $s$ and each configuration $x$. Then the PCA dynamics is ergodic, where its unique stationary measure is $\mathsf{Bern}\left(\frac{1}{2} \right)^{\otimes \mathbb{T}^{2}}$. 
\end{enumerate}
\end{teo}

\begin{teo}\label{Th:Ergo_2}
\label{axe3} Let $d=3$ and let the transition probabilities be symmetric under spin-flip.
Denote by $\alpha:=p_{\mathbf{o}}(+|+++)$ and $\gamma:=p_{\mathbf{o}}(+|-++) + p_{\mathbf{o}}(+|+-+) + p_{\mathbf{o}}(+|++-)$. 
\noindent
Then the PCA transition rule is 
\begin{enumerate}
\item[(a)] ergodic, if $\alpha$ and $\gamma$ satisfy
\begin{enumerate}
\item[(i)] $1+\alpha-\gamma = 0$, or 
\item[(ii)] the PCA dynamics is attractive and $1+\alpha-\gamma \neq 0$ and $3\alpha+\gamma\leq 5$, or
\item[(iii)] the PCA dynamics is repulsive and $1+\alpha-\gamma \neq 0$ and $3\alpha+\gamma\geq 1$. 
\end{enumerate}
In this case the unique stationary measure is given by $\mathsf{Bern}\left(\frac{1}{2} \right)^{\otimes\mathbb{T}^3}$.
\item[(b)] non-ergodic, if $\alpha$ and $\gamma$ satisfy
\begin{enumerate}
\item[(i)] $1+\alpha-\gamma \neq  0$ and $3\alpha+\gamma > 5$. In this case, 
we have several stationary Bernoulli product measures  with parameter \[\mathsf{p} \in \left \{ \frac{1}{2},  \frac{1 + \sqrt{1+\frac{4(1-\alpha)}{1+\alpha-\gamma}}}{2}, \frac{\sqrt{1-\frac{4(1-\alpha)}{1+\alpha-\gamma}}}{2} \right \}, \]
or 
\item[(ii)] the PCA dynamics is repulsive and $1+\alpha-\gamma \neq  0$ and $3\alpha+\gamma < 1$.
\end{enumerate}
\end{enumerate}
\end{teo}

\medskip

\begin{remark}
In the last case (Theorem \ref{Th:Ergo_2} $(b)$-$(ii)$),  we can actually prove that the PCA oscillates between two Bernoulli product measures with distinct parameters $\mathsf{p}$. Further details are presented in Section \ref{proof:repulsive}.
\end{remark}


\noindent
Before we pass to the proofs of the theorems we will discuss some examples.

\subsubsection{Example 1}
\label{Ergo_2}
For $d=3$ and $\beta>0$, let us consider the PCA with transition probabilities given by
\begin{equation}\label{transprob}
p_i(s|x) = \frac{1}{2} \left ( 1 + s \tanh \left (\beta \sum_{k=0}^2 J_k x_{i+e_{k}}\right )\right)	
\end{equation}
where $J_0,J_1$ and $J_2 \in \mathbb{R}$. 
Hence, for suitable values of the constants, there exists a critical $\beta_c\in(0,\infty)$ such that the PCA is ergodic for 
$\beta \leq \beta_c$ and non-ergodic otherwise. In fact the following result holds.


\begin{prop}\label{prop:Ex1}
Suppose that one of the following conditions on the coupling constants $J_0, J_1, J_2$ is fulfilled.
\begin{itemize}
\item[(C1)] $J_0,J_1,J_2 >0$ and $J_{0} \leq J_{1} + J_{2}$, $J_{1} \leq J_{0} + J_{2}$, and $J_{2} \leq J_{0} + J_{1}$.
\item[(C2)] $J_0,J_1,J_2< 0$ and $J_{0} \geq J_{1} + J_{2}$, $J_{1} \geq J_{0} + J_{2}$, and $J_{2} \geq J_{0} + J_{1}$.
\end{itemize}

Let $\alpha, \gamma$ be defined as in Theorem \ref{Th:Ergo_2}, and let function $f:\mathbb{R}_+\rightarrow \mathbb{R}$ be defined as
\[
f(\beta)=3\alpha + \gamma.
\]
Then, there exists $\beta_c \in (0,\infty)$ 
 depending on the constants $J_0,J_1,J_2$ such that for
\begin{itemize}
\item[(a)] $\beta \leq \beta_c$ the PCA dynamics associated to the local transition probabilities given by \eqref{transprob} is ergodic, and
\item[(b)] $\beta > \beta_c$ the dynamics is non-ergodic.
\end{itemize}
\end{prop}

\begin{remark2}
 Note that, thanks to the spin-flip symmetry of the probabilities (\ref{transprob}), we can apply Theorem~\ref{Th:Ergo_2}. Moreover, we remark that the lattice model equivalent to (\ref{transprob})
has been extensively studied  in  \cite{louis}.
\end{remark2}
\begin{remark2}
If condition $(C1)$ holds, then $\beta_c=f^{-1}(5)$. Otherwise, if $(C2)$ holds, then $\beta_c=f^{-1}(1)$. In particular, if $J_0=J_1=J_2=J\in\mathbb{R} \backslash \{0\}$, it follows that $\beta_c=\frac{1}{2|J|}\log(1+2^{2/3})$. In \cite{louis2} a similar ferromagnetic PCA has been studied on $\mathbb{Z}^d$ where in the particular case $d=2$ the value of $\beta_c$ is given by $\beta_c = \frac{1}{2J}\log(1+\sqrt{2})$.
\end{remark2}
\medskip
\subsubsection{Example 2}
\label{Ergo_4}
Let us consider the PCA on the $3$-ary tree defined as follows. Suppose that at each step every spin assume the value corresponding to the majority
among their children. After that each spin make an error with a probability $\epsilon \in (0,1)$ independently of each other, that is, if the spin at the site $i$ 
assumed the value $+1$ (resp. $-1$), then it will change to $-1$ (resp. $+1$) with probability $\epsilon$ and keep
the value $+1$ (resp. $-1$) with probability $1-\epsilon$. Note that such a system follows a CA dynamics, namely the majority rule, with the addition of a noise. For a more detailed study of this kind of PCAs, see \cite{MST}.

In the example described above, we have
\[
p_{\mathbf{o}}(+|+++)=p_{\mathbf{o}}(+|++-)=p_{\mathbf{o}}(+|+-+)=p_{\mathbf{o}}(+|-++)=1-\epsilon. 
\]
This PCA  has been first studied in \cite{dobrushin}, where non-ergodicity has been proven only for sufficiently small $\epsilon$. In the next proposition we fully characterize its behavior for the whole range of $\epsilon$.

\begin{prop}\label{prop:Ex2}
There exist two critical values $\epsilon_{c}^{(1)}=\frac{1}{6}$
and $\epsilon_{c}^{(2)}=\frac{5}{6}$ such that for every $\epsilon \in (0,1)$ 
\begin{itemize}
\item[(a)] the PCA dynamics is ergodic if $\epsilon_{c}^{(1)} \leq \epsilon \leq \epsilon_{c}^{(2)}$, and
\item[(b)] non-ergodic for $\epsilon \notin [\epsilon_{c}^{(1)}, \epsilon_{c}^{(2)}]$.
\end{itemize}
\end{prop}

\section{Discussion}
In this work we proved the correspondence between stationary measures for PCAs on infinite rooted trees and time-invariant Gibbs measures for a corresponding statistical mechanical model. As mentioned before, the proof of such correspondence is very general and can be applied for any PCA on a (locally finite) infinite rooted tree with finite single spin space $S$. The main implication of this fact is once we establish conditions for uniqueness of Gibbs measures for such a system, we guarantee the uniqueness of stationary distributions for the associated PCA, compare also \cite{Steiff}. On the other hand, the existence of multiple stationary measures implies on the phase transition in the statistical mechanical model. In this way
we provide a partial relationship between ergodicity and phase transition extending the results from \cite{pca}.

Restricting to the study of PCAs on a $d$-ary tree $\mathbb{T}^{d}$  with translation-invariant local transition probabilities with single spin space $S= \{-1,+1\}$, we were able to find ergodicity properties for such class of PCAs. The assumption that the choice of a local transition probability at a site $i$ only depends upon the values of the spins of the children of $i$ allowed us to derive several important properties, for instance, equations (\ref{prodoido}) and (\ref{bodybuilder2}). Equation (\ref{prodoido}) shows us that the probability distributions of the spins at time $n$ are independent, such fact lead us to characterize the stationary measures of such a system whose form are product measures. In this way, we naturally obtained a polynomial function $F$ defined on the interval $[0,1]$ whose expression is given by
  \begin{equation}
    \label{eq:functionF}
    F(\mathsf{p}) = \sum\limits_{l=0}^{d}(-1)^{l}\left[\sum\limits_{\substack{I \subseteq \{0,\dots,d-1\} \\ |I| = l}}\left(\sum\limits_{\substack{\xi \in \{-1,+1\}^{d} \\ \xi_{k}= -1\; \textnormal{for all}\; k \in I}}(-1)^{\#\{m:\xi_{m}=-1\}}p_{\mathbf{o}}(+1|\xi)\right)\right]\mathsf{p}^{d-l} 
  \end{equation}
such that, according to equation (\ref{bodybuilder}), a Bernoulli product measure with parameter $\mathsf{p}$ is a stationary measure for the PCA dynamics if and only if $\mathsf{p}$ is a fixed point of $F$. Furthermore, based on equation (\ref{bodybuilder2}), the convergence of $P^n(x,\{y_{\mathbf{o}}= +1\})$ for a shift-invariant configuration $x$ (that is, for $x$ that satisfies $\theta_i x = x$ for all $i$) can be studied in terms on the behavior of the iterations $F^n$, since the identity  $P^{n+1}(x,\{y_{\mathbf{o}}= +1\}) = F(P^n(x,\{y_{\mathbf{o}}= +1\}))$ holds.

  We applied the techniques described above in the cases where $d=1,2$, and $3$. For $d = 1$, we the PCA dynamics is ergodic and the unique stationary measure is a Bernoulli product measure with parameter $\mathsf{p}$ given by equation (\ref{1dstat}). Note that this case is equivalent to the study of a PCA on $\mathbb{N}$ where the choice of the value of the spin located at $i$ at time $n+1$  depends only on the value of the spin at $i+1$ at time $n$. Extensions of this result where $U(i)= \{i,i+1\}$ were extensively studied in \cite{dobrushin}, moreover, more recent generalizations that considers one-dimensional PCAs with general finite alphabets and characterizations of Markov stationary measures can be found in \cite{Casse1,Casse2}. For the cases $d=2$ and $d=3$ we assumed the invariance of the local transition probabilities under spin-flip in order to guarantee the existence of a stationary Bernoulli product measure (which has parameter $\frac{1}{2}$). Under this restriction, we obtained a full characterization the dynamics of PCAs with $d=2$, and $d=3$ with the additional hypothesis of attractiveness (resp. repulsiveness).

  For further generalizations, in order to drop the assumption of spin-flip symmetry and extend the results for any $d$, it is necessary to investigate the general properties of the polynomial function $F$ regarding its fixed points and the behavior of its iterates $F^n$. It is also worth investigating generalizations of PCAs from Examples 1 and 2. Note that Theorem \ref{stm} together with Dobrushin's uniqueness theorem implies that  for a PCA on $\mathbb{T}^d$ whose local transition probabilities are given by
\begin{equation}\label{transprob2}
p_i(s|x) = \frac{1}{2} \left ( 1 + s \tanh \left (\beta \sum_{k=0}^{d-1} J_k x_{i+e_{k}}\right )\right)	
\end{equation}
there is a unique stationary measure given by $\mathsf{Bern}(\frac{1}{2})^{\otimes \mathbb{T}^d}$ for $\beta$  small  enough, suggesting the ergodicity at high temperatures.

Another kind direction that should be considered in the future is the possibility of inclusion of finite alphabets other that $S = \{-1,+1\}$ and the possibility of influence of the state at the vertex $i$ at time $n$ on its state at time $n+1$, more precisely, the possibility of considering $U(i) = \{i, i+e_{\mathbf{o}}, \dots, i + e_{d-1}\}$. Such assumptions require a new approach once equations (\ref{prodoido}), (\ref{bodybuilder}) and (\ref{bodybuilder2}) would no longer be valid, so, one possible direction that should be chosen would be towards an extension of the results from \cite{Casse1,Casse2}.

\section{Proofs of Ergodicity results}\label{proofs}


\subsection{Proof of Theorem \ref{Ergo_1}}
\subsubsection{Case (a)}
\begin{proof}[]
Note that a  PCA on $\mathbb{T}^1$ is equivalent to a PCA model on $\mathbb{Z}_+$.
In order to simplify the computations, let us use $a$ and $b$ to denote $p_{\mathbf{o}}(+|+)$ and $p_{\mathbf{o}}(+|-)$, respectively. 
Since the local transition probabilities have positive rates, then, we have  $|a- b| < 1$. It follows that for each point $x$ in $\Omega_{0}$, we have
\begin{eqnarray*}
P^{n+1}(x,\{y_{\mathbf{o}} = +1\}) &=& \int P(z,\{y_{\mathbf{o}} = +1\}) P^{n}(x,dz) \\
&=& a\cdot P^{n}(x,\{y_{e_{0}} = +1\}) + b\cdot P^{n}(x,\{y_{e_{0}} = -1\}) \\
&=& (a - b)\cdot P^{n}(x,\{y_{e_{0}} = +1\}) + b \\
&=& (a - b)\cdot P^{n}(\Theta_{e_{0}}x,\{y_{\mathbf{o}} = +1\}) + b
\end{eqnarray*} 
for each positive integer $n$. Note that the relation above can also be obtained by means of equation (\ref{bodybuilder2}). Thus, the quantity above
can be expressed as
\begin{equation*}
P^{n}(x,\{y_{\mathbf{o}} = +1\}) = (a - b)^{n-1}\cdot p_{\mathbf{o}}(+1|\Theta_{\underbrace{e_{0}+\cdots+{e_{0}}}_{n-1\, \textnormal{times}}}x) + b \cdot \sum\limits_{k=0}^{n-2} (a - b)^{k}.	
\end{equation*}
It follows that for any initial configuration $x$, the probability $P^{n}(x,\{y_{\mathbf{o}} = +1\})$ converges to $\mathsf{p} = \frac{b}{1-(a - b)}$ as $n$ approaches infinity.
Therefore, using equation (\ref{prodoido}), we conclude that this PCA is ergodic, where its unique attractive stationary measure is $\mathsf{Bern(p)}^{\otimes \mathbb{T}^{1}}$.
\end{proof}

\subsubsection{Case (b)}\label{proof:caseb}
\begin{proof}
Let $a,b\in(0,1)$ defined by $a = p_{\mathbf{o}}(+|--) = 1 - p_{\mathbf{o}}(+|++)$ 
and $b = p_{\mathbf{o}}(+|-+) = 1 - p_{\mathbf{o}}(+|+-)$, respectively.  
Let us show that $\mathsf{Bern(\frac{1}{2})}^{\otimes \mathbb{T}^2}$, in fact, is the
unique attractive stationary measure, that is, for every initial configuration $x$ we have $P^n(x,\cdot) \rightarrow \mathsf{Bern(\frac{1}{2})}^{\otimes \mathbb{T}^{2}}$ as $n$ approaches infinity. According to equation (\ref{bodybuilder2}), we have
\begin{eqnarray*}
P^{n+1}(x,\{y_{\mathbf{o}}=+1\}) = (1-b-a) P^{n}(\Theta_{e_{0}}x,\{y_{\mathbf{o}}=+1\}) + (b- a) P^{n}(\Theta_{e_{1}}x,\{y_{\mathbf{o}}=+1\}) + a.	
\end{eqnarray*}
By induction, we can show that
\begin{eqnarray*}
P^{n}(x,\{y_{\mathbf{o}} = +1\}) &=& \sum\limits_{i \in \{0,1\}^{n-1}}(1-b -a)^{\#\{k:i_{k} = 0\}}(b -a)^{\#\{k:i_{k} = 1\}}P(\Theta_{i}x,\{y_{\mathbf{o}} = +1\}) \\
&& + a \sum_{l=0}^{n-2}\sum\limits_{i \in \{0,1\}^{l}}(1-b -a)^{\#\{k:i_{k} = 0\}}(b -a)^{\#\{k:i_{k} = 1\}}.
\end{eqnarray*}
Using the fact that for any real numbers $p$ and $q$, the relation
\[\sum\limits_{i \in \{0,1\}^{l}}p^{\#\{k:i_{k} = 0\}}q^{\#\{k:i_{k} = 1\}} = (p+q)^{l}\]
holds for every nonnegative integer $l$, it follows that 
\begin{equation}\label{arehoo}
P^{n}(x,\{y_{\mathbf{o}} = +1\}) = \sum\limits_{i \in \{0,1\}^{n-1}}(1-b -a)^{\#\{k:i_{k} = 0\}}(b -a)^{\#\{k:i_{k} = 1\}}P(\Theta_{i}x,\{y_{\mathbf{o}} = +1\}) 
+ a \sum_{l=0}^{n-2}(1-2a)^{l}.
\end{equation}
Since the absolute value of the first term of equation (\ref{arehoo}) is bounded by
\[\sum\limits_{i \in \{0,1\}^{n-1}}|1-b -a|^{\#\{k:i_{k} = 0\}}|b -a|^{\#\{k:i_{k} = 1\}} = (|1-b -a|+|b -a|)^{n-1},\]
then
\begin{equation}
\lim_{n \to \infty}P^{n}(x,\{y_{\mathbf{o}} = +1\})  = 	a \sum_{l=0}^{\infty}(1-2a)^{l} = \frac{1}{2}.
\end{equation}
Therefore, by means of equation (\ref{prodoido}), we conclude that  
$\mathsf{Bern(}\frac{1}{2}\mathsf{)}^{\otimes \mathbb{T}^{2}}$ is  the unique attractive stationary measure
of the PCA. 
\end{proof}
\subsection{Proof of Theorem \ref{Th:Ergo_2}}

\subsubsection{Case (a)-(i) and (b)-(i)}
\begin{proof}
Recall we abbreviated  $\alpha = p_{\mathbf{o}}(+|+++)$ and $\gamma=  p_{\mathbf{o}}(+|-++)+ p_{\mathbf{o}}(+|+-+)+ p_{\mathbf{o}}(+|++-)$. 
From Lemma \ref{lem:statsol} we know that a stationary product measure has to satisfy the condition
\begin{equation}\label{vemmonstro2}
\int p_{\mathbf{o}}(+1|x) \nu(dx) = \mathsf{p}
\end{equation}
which was equivalent to solving equation \eqref{bodybuilder}, i.e.

\begin{equation}\label{eq:f}
2(1+\alpha-\gamma) \mathsf{p}^{3}-3(1+\alpha- \gamma) \mathsf{p}^{2}+(3\alpha-\gamma-1) \mathsf{p} + (1-\alpha) = 0.
\end{equation}
Since $\mathsf{p}=\frac{1}{2}$ is a solution for the equation above, then, it can be written as
\begin{equation}\label{nabo}
2\left(\mathsf{p}-\frac{1}{2}\right) \left[(1+\alpha-\gamma)\mathsf{p}^{2} - (1+\alpha-\gamma)\mathsf{p} - (1-\alpha) \right] = 0.	
\end{equation} 

\noindent
Suppose that $1+\alpha-\gamma = 0$. 
Then, analogously as in the previous case, we have
\begin{eqnarray*}
&&P^{n}(x,\{y_{\mathbf{o}} = +1\}) \\ 
\\
&&= \hspace{-0.5cm}\sum\limits_{i \in \{0,1,2\}^{n-1}}\hspace{-0.5cm}(\alpha-p_{\mathbf{o}}(+|++-))^{\#\{k:i_{k} = 0\}}(\alpha-p_{\mathbf{o}}(+|+-+))^{\#\{k:i_{k} = 1\}}(\alpha-p_{\mathbf{o}}(+|-++))^{\#\{k:i_{k} = 2\}}P(\Theta_{i}x,\{y_{\mathbf{o}} = +1\}) \\
&&\quad	 + (1-\alpha) \sum_{l=0}^{n-2}\sum\limits_{i \in \{0,1,2\}^{l}}(\alpha-p_{\mathbf{o}}(+|++-))^{\#\{k:i_{k} = 0\}}(\alpha-p_{\mathbf{o}}(+|+-+))^{\#\{k:i_{k} = 1\}}(\alpha-p_{\mathbf{o}}(+|-++))^{\#\{k:i_{k} = 2\}}.
\end{eqnarray*}
The equation above implies that $P^{n}(x,\{y_{\mathbf{o}} = +1\}) \to \frac{1}{2}$ as $n$ approaches infinity, therefore, by means of the same argument as used in Section \ref{proof:caseb}, we conclude that the dynamics is ergodic.

Now, if $1+\alpha-\gamma \neq 0$, we have two other solutions
\begin{equation}\label{pemais}
\mathsf{p_{+}} = \frac{1+\sqrt{1+\frac{4(1-\alpha)}{1+\alpha-\gamma}}}{2}	
\end{equation}
and
\begin{equation}\label{pemenos}
\mathsf{p_{-}} = \frac{1-\sqrt{1+\frac{4(1-\alpha)}{1+\alpha-\gamma}}}{2}.	
\end{equation}
Therefore, both  $\mathsf{p_{-}}$ and $\mathsf{p_{+}}$ are inside the interval $(0,1)$ and are different from $\frac{1}{2}$ if and only if 
$3\alpha + \gamma > 5$. 
\end{proof}

\subsubsection{Case (a)-(ii)}
\begin{proof}
Let us consider a PCA with attractive dynamics. Again, by using Lemma~\ref{lem:statsol}, we can find a map $F:[0,1] \rightarrow \mathbb{R}$ \begin{equation}\label{fidelcastro}
F(\pp) = 2(1+\alpha-\gamma)\pp^{3}-3(1+\alpha-\gamma)\pp^{2}+(3\alpha-\gamma)\pp+(1-\alpha){}
\end{equation}
such that its fixed points correspond to the parameters of the stationary Bernoulli product measures.
We will show that $F$ has a unique attractive fixed point
at $\pp=\frac{1}{2}$, that is, such fixed point satisfies $F^n(\mathsf{q}) \to \pp$ as $n$ approaches infinity for any point $\mathsf{q} \in [0,1]$.
Let us prove that $F$ is an increasing function that satisfies 
\begin{eqnarray}\label{malandramente}
\begin{cases}
F(\pp) > \pp &\text{for all}\; \pp < \frac{1}{2}, \\
F(\frac{1}{2}) = \frac{1}{2} &\text{and} \\
F(\pp) < \pp &\text{for all}\; \pp > \frac{1}{2}.
\end{cases}	
\end{eqnarray}
Suppose that $1+\alpha-\gamma <0$.  Due to the attractiveness of the dynamics, it follows that $3\alpha \geq \gamma$ and
the minimum value of $F'$ given by $F'(0)=F'(1) = 3\alpha-\gamma$ is nonnegative. Therefore, $F$ is increasing. 
Moreover, the property (\ref{malandramente}) follows from the identity
\begin{equation}
F(\pp)-\pp = 2\left(\mathsf{p}-\frac{1}{2}\right) \left[(1+\alpha-\gamma)\mathsf{\pp}^{2} - (1+\alpha-\gamma)\mathsf{\pp} - (1-\alpha) \right] 	
\end{equation}
where $(1+\alpha-\gamma)\mathsf{\pp}^{2} - (1+\alpha-\gamma)\mathsf{\pp} - (1-\alpha)<0$  for all $\pp \neq \frac{1}{2}$. 
Now, let us consider the case where $1+\alpha-\gamma > 0$. The attractiveness of the dynamics implies that
$\gamma \geq 3(1-\alpha)$, so, the minimum value of $F'$ is $F'(\frac{1}{2})= (-3+3\alpha+\gamma)/2  \geq 0$. Again, we prove that $F$ is increasing. 
Furthermore, we have (\ref{malandramente}) by means of the equation
\begin{equation}
F(\pp)-\pp = 2\left(\mathsf{\pp}-\frac{1}{2}\right)(1+\alpha-\gamma)(\pp-\mathsf{\pp}_{-})(\pp-\mathsf{p}_{+}) 	
\end{equation}
where $\mathsf{\pp}_{-}<0$ and $\mathsf{\pp}_{+}>1$ are given by equation (\ref{pemenos}) and (\ref{pemais}), respectively. 
Since $F$ is increasing, $F(0)=1-\alpha < \frac{1}{2}$ and $F(1)=\alpha > \frac{1}{2}$, then $F(\pp)$ belongs to $\left[1-\alpha,\frac{1}{2} \right) \subseteq \left[0,\frac{1}{2} \right)$ for all $\pp$ in $\left[0,\frac{1}{2} \right)$ and 
$F(\pp)$ belongs to $\left(\frac{1}{2},\alpha \right] \subseteq \left(\frac{1}{2},1 \right]$ for all $\pp$ in $\left(\frac{1}{2},1\right]$. Using the continuity of $F$, we easily conclude that
$\lim_{n \to \infty}F^{n}(\mathsf{q})=\frac{1}{2}$ for every point $\mathsf{q}$ that belongs to the interval $[0,1]$, therefore, $\pp=\frac{1}{2}$ is the unique attractive fixed point for $F$.

It follows from  equation (\ref{bodybuilder2})  that
\[P^{n+1}(x_{-},\{y_{\mathbf{o}}=+1\}) = F(P^{n}(x_{-},\{y_{\mathbf{o}}=+1\}))\]
and
\[P^{n+1}(x_{+},\{y_{\mathbf{o}}=+1\}) = F(P^{n}(x_{+},\{y_{\mathbf{o}}=+1\})),\]
where $x_{-}$ and $x_{+}$ are respectively the configurations with all spins $-1$ and $+1$ on $\mathbb{T}^3$. The conclusion above implies that both $P^{n}(x_{-},\{y_{\mathbf{o}}=+1\})$ and $P^{n}(x_{+},\{y_{\mathbf{o}}=+1\})$ converge to
$\frac{1}{2}$ as $n$ approaches infinity.
Therefore,
since the inequality $x_{-}\leq x\leq x_{+}$ holds for every configuration $x$, it follows from  Definition~\ref{def_attr}
 that
\begin{equation}
P^{n}(x_{-},\{y_{\mathbf{o}} = +1\}) \leq P^{n}(x,\{y_{\mathbf{o}} = +1\})\leq P^{n}(x_{+},\{y_{\mathbf{o}} = +1\}),		
\end{equation}
therefore,
\begin{equation}
\lim_{n \to \infty}P^{n}(x,\{y_{\mathbf{o}} = +1\}) =\frac{1}{2}.
\end{equation}
Finally, we conclude that the probability $P^{n}(x,\cdot)$ converges to $\mathsf{Bern(\frac{1}{2})}^{\otimes \mathbb{T}^3}$ as $n$ approaches infinity, independently on the initial configuration $x$, hence, the PCA dynamics is ergodic.
\end{proof}
\subsubsection{Case (a)-(iii) and (b)-(ii)}\label{proof:repulsive}
\begin{proof}
Let us consider a new PCA described by a probability kernel $Q$ defined by
\begin{equation}
Q(dy|x) = \bigotimes\limits_{i \in \mathbb{T}^3} q_{i}(dy_{i}|x),	
\end{equation}
where each probability $q_{i}$ is given by 
\begin{equation}
q_{i}(\,\cdot\,|x) = p_{i}(\,\cdot\,|-x).	
\end{equation}	
It is easy to see that this PCA satisfies the spin-flip condition. In the case where we have $3\alpha+\gamma \geq 1$, 
if we consider $\alpha'$ and $\gamma'$ respectively defined by $\alpha'=q_{\mathbf{o}}(+|+++)$ and
$\gamma'=q_{\mathbf{o}}(+|++-)+q_{\mathbf{o}}(+|+-+)+q_{\mathbf{o}}(+|-++)$, then we have 
\[1+\alpha'-\gamma' = -(1+\alpha-\gamma) \neq 0,\]
and
\[3\alpha'+\gamma'= 6 - (3\alpha+\gamma)\leq 5.\]
Therefore, in this case the PCA dynamics described by $Q$ is ergodic. It is easy to check that $P^{n}(x,\cdot) = Q^{n}((-1)^{n}x,\cdot)$ holds for every positive integer $n$
and each configuration $x$. Therefore, the ergodicity of $P$ follows.	

In order to prove the non-ergodicity for the case $3\alpha +\gamma <1$, let us consider again the function $F:[0,1] \rightarrow \mathbb{R}$ given by 
equation (\ref{fidelcastro}).
 It is straightforward to show that
  \begin{equation}
    F(\pp) - (1-\pp) = 2(1-\alpha-\gamma)\left(\pp - \frac{1}{2}\right)( 1 - \mathsf{q_{-}})( 1 - \mathsf{q_{+}}),
  \end{equation}
where  $\mathsf{q_{-}}$ and $\mathsf{q_{+}}$ are the elements in the interval $(0,1)$ given by
\begin{equation}\label{qemais}
\mathsf{q_{-}} = \frac{1+\sqrt{1-\frac{4\alpha}{1+\alpha-\gamma}}}{2}	
\end{equation}
and
\begin{equation}\label{qemenos}
\mathsf{q_{+}} = \frac{1-\sqrt{1-\frac{4\alpha}{1+\alpha-\gamma}}}{2},	
\end{equation}  
respectively. It follows that
\begin{equation}
  \begin{cases}
    F(\pp) < 1 - \pp & \text{if $\pp \in [0,\mathsf{q}_{-})$,}\\
    F(\pp) > 1 - \pp & \text{if $\pp \in (\mathsf{q}_{-},\frac{1}{2})$,}\\
    F(\pp) < 1 - \pp & \text{if $\pp \in (\frac{1}{2},\mathsf{q}_{+})$, and}\\
    F(\pp) > 1 - \pp & \text{if $\pp \in (\mathsf{q}_{+},1].$}
  \end{cases}
\end{equation} 
Because of the repulsiveness of the dynamics, we have $3\alpha-\gamma \leq 0$ and $F'(\frac{1}{2}) = \frac{1}{2}(-3+3\alpha+\gamma)<-1$, thus, $F$ is a decreasing function. In addition, we have $F(\pp) = 1-F(1-\pp)$ for every $\pp$ in $[0,1]$. So, we obtain
\begin{equation}
  \begin{cases}
    \pp < F^2(\pp) < \mathsf{q}_{-} & \text{if $\pp \in [0,\mathsf{q}_{-})$,}\\
    \mathsf{q}_{-} < F^2(\pp) < \pp & \text{if $\pp \in (\mathsf{q}_{-},\frac{1}{2})$,}\\
    \pp < F^2(\pp) < \mathsf{q}_{+}  & \text{if $\pp \in (\frac{1}{2},\mathsf{q}_{+})$, and}\\
    \mathsf{q}_{+} < F^2(\pp) < \pp & \text{if $\pp \in (\mathsf{q}_{+},1].$}
  \end{cases}
\end{equation}
Therefore, we conclude that
\begin{equation}\label{osc1}
\lim_{n \to \infty}F^{2n}(\pp) = 
  \begin{cases}
    \mathsf{q}_{-} &\text{if $\pp \in [0,\frac{1}{2})$, and}\\
    \mathsf{q}_{+} &\text{if $\pp \in (\frac{1}{2},1]$;}
  \end{cases}
\end{equation}
similarly, we also have
\begin{equation}\label{osc2}
\lim_{n \to \infty}F^{2n+1}(\pp) = 
  \begin{cases}
    \mathsf{q}_{+} &\text{if $\pp \in [0,\frac{1}{2})$, and}\\
    \mathsf{q}_{-} &\text{if $\pp \in (\frac{1}{2},1]$.}
  \end{cases}
\end{equation}
Thus, we finally conclude that, by means of equations (\ref{bodybuilder2}), (\ref{osc1}) and (\ref{osc2}), the probabilities 
$P^{2n+1}(x_{+},\cdot)$ and  $P^{2n}(x_{+},\cdot)$ converge to $\mathsf{Bern(}\mathsf{q_{-}}\mathsf{)}^{\otimes \mathbb{T}^3}$ and $\mathsf{Bern(}\mathsf{q_{+}}\mathsf{)}^{\otimes \mathbb{T}^3}$, respectively,  as $n$ approaches infinity. So, the PCA dynamics is not ergodic. 

\end{proof}
\subsection{Proof of Proposition~\ref{prop:Ex1} }
\begin{proof}
The PCA is fully described by the numbers
\begin{eqnarray*}
p_{\mathbf{o}}(+|+++) &=& \frac{1}{2}\left(1+\tanh\beta(J_{0}+J_{1}+J_{2})\right),\\
p_{\mathbf{o}}(+|++-) &=& \frac{1}{2}\left(1+\tanh\beta(J_{0}+J_{1}-J_{2})\right),\\
p_{\mathbf{o}}(+|+-+) &=& \frac{1}{2}\left(1+\tanh\beta(J_{0}-J_{1}+J_{2})\right),	
\end{eqnarray*}
and
\begin{equation*}
p_{\mathbf{o}}(+|-++) = \frac{1}{2}\left(1+\tanh\beta(-J_{0}+J_{1}+J_{2})\right).	
\end{equation*}
Note that assumption (C1) from Example 1 implies that $J_{0}+J_{1}-J_{2}<J_{0}+J_{1}+J_{2}$, $J_{0}-J_{1}+J_{2}<J_{0}+J_{1}+J_{2}$, and $-J_{0}+J_{1}+J_{2}<J_{0}+J_{1}+J_{2}$; and
at most one of the quatities $J_{0}+J_{1}-J_{2}$, $J_{0}-J_{1}+J_{2}$ and 
$-J_{0}+J_{1}+J_{2}$ can be equal zero.
Therefore, the map $g:\mathbb{R} \rightarrow \mathbb{R}$ given by
\begin{eqnarray*}
g(\beta)&=&1+\alpha-\gamma \\
&=& \frac{1}{2}(\tanh\beta(J_{0}+J_{1}+J_{2})-\tanh\beta(J_{0}+J_{1}-J_{2})-\tanh\beta(J_{0}-J_{1}+J_{2})-\tanh\beta(-J_{0}+J_{1}+J_{2}))
\end{eqnarray*}
satisfies $g(0)=0$ and
\begin{eqnarray*}
g'(\beta) &=& \frac{1}{2}\Bigg(\frac{J_{0}+J_{1}+J_{2}}{\cosh^{2}\beta(J_{0}+J_{1}+J_{2})}-\frac{J_{0}+J_{1}-J_{2}}{\cosh^{2}\beta(J_{0}+J_{1}-J_{2})} 
-\frac{J_{0}-J_{1}+J_{2}}{\cosh^{2}\beta(J_{0}-J_{1}+J_{2})}\\
&&\quad -\frac{-J_{0}+J_{1}+J_{2}}{\cosh^{2}\beta(-J_{0}+J_{1}+J_{2})}\Bigg)\\
&<& \frac{1}{2\cosh^{2}\beta(J_{0}+J_{1}+J_{2})}\left((J_{0}+J_{1}+J_{2}) - (J_{0}+J_{1}-J_{2})-(J_{0}-J_{1}+J_{2})-(-J_{0}+J_{1}+J_{2})\right) =0.	
\end{eqnarray*}
It follows that $g(\beta)=1+\alpha-\gamma <0$ for all $\beta>0$.{}
Moreover, note that the function $f:\mathbb{R} \rightarrow \mathbb{R}$ given by 
\begin{eqnarray*}
f(\beta) &=& 3\alpha + \gamma\\
&=& 3 + \frac{3}{2}\tanh\beta(J_{0}+J_{1}+J_{2}) + \frac{1}{2}(\tanh\beta(J_{0}+J_{1}-J_{2})+\tanh\beta(J_{0}-J_{1}+J_{2})+\tanh\beta(-J_{0}+J_{1}+J_{2}))
\end{eqnarray*}
is increasing, $f(0) = 3$, and $\lim_{\beta \to \infty}f(\beta) \geq 5 + \frac{1}{2}$. It follows that there is a unique positive real number $\beta_{c}$	that satisfies $f(\beta_{c})=5$.
Since this PCA dynamics satisfies the spin-flip property and is attractive, according to Theorem \ref{Th:Ergo_2}, the PCA is ergodic for $\beta \leq \beta_{c}$ and non-ergodic for $\beta > \beta_{c}$.

Since we proved the result considering the case where condition $(C1)$ holds, the proof for the case $(C2)$ is straightforward. 
\end{proof}
\subsection{Proof of Proposition \ref{prop:Ex2}}
\begin{proof}
Clearly the PCA satisfies the spin-flip property. Note that in both cases we have $1+\alpha-\gamma = 2\epsilon -1$. It follows that the PCA is ergodic 
for $\epsilon = \frac{1}{2}$. Furthermore, note that the PCA is attractive for $0< \epsilon < \frac{1}{2}$, repulsive for $\frac{1}{2}<\epsilon<1$, and
in both cases we have $1+\alpha-\gamma \neq 0$. 

Let us suppose that $\epsilon \in (0,\frac{1}{2})$. Since $3\alpha+\gamma=6(1-\epsilon)$, it follows from Theorem \ref{Th:Ergo_2} that the PCA is non-ergodic
for $\epsilon < \frac{1}{6}$ and ergodic for $\frac{1}{6} \leq \epsilon < \frac{1}{2}$.	Now, if $\epsilon \in (\frac{1}{2},1)$, then again by Theorem \ref{Th:Ergo_2}, the PCA is ergodic for $\frac{1}{2}<\epsilon \leq \frac{5}{6}$ and non-ergodic for $\frac{5}{6}< \epsilon < 1$. 
\end{proof}


\appendix
\section{Appendix}\label{app1}
\subsection{Proof of Theorem \ref{stm}}
\setcounter{equation}{0}
\renewcommand{\theequation}{A\thechapter.\arabic{equation}}
Before we follow to the proof of Theorem \ref{stm} it will be convenient to construct a special sequence $(\Delta_{n})_{n \in \mathbb{N}}$ of subsets of $\mathbb{Z}\times \vv$. Given a 
positive integer $n$ and a nonempty finite subset $F$ of $\mathsf{V}$, let us define a subset $\Delta(n,F)$ of $\mathbb{Z}\times \vv$ as follows.
Let $\Lambda_{n}$ be the set given by
\[\Lambda_{n} = \{(n,i) : i \in F\},\]
and for each integer $m < n$ let 
\[\Lambda_{m} =\bigcup\limits_{\mathsf{x} \in \Lambda_{m+1}}U(\mathsf{x})\cup\{(m,i): i\in F\}\]
Then, we define $\Delta(n,F)$ by
\[\Delta(n,F) = \bigcup\limits_{m=-n}^{n}\Lambda_{m}.\]

\begin{remark}\label{rmks}
Observe that
\begin{enumerate}[label=(\alph*),ref=\alph*]
\item[\textnormal{(}a\textnormal{)}] $\Delta(n,F)$ is a finite subset of $\mathbb{Z}\times \vv$,
\item[\textnormal{(}b\textnormal{)}] we have $\{-n,\dots,0,\dots,n\}\times F \subseteq \Delta(n,F) \subseteq \{-n,\dots,0,\dots,n\}\times \mathsf{V}$, and
\item[\textnormal{(}c\textnormal{)}] for every point $\mathsf{x}$ in $\Delta(n,F)$, if $\pi_{\mathbb{Z}}(\mathsf{x}) \neq -n$, then $U(\mathsf{x}) \subseteq \Delta(n,F)$. 
\end{enumerate}
\end{remark}

Now, if $\varphi$ is a one-to-one function from $\mathbb{N}$ onto $\mathsf{V}$, then let
\begin{equation}
\Delta_{1} = \Delta(1,\{\varphi(1)\}),
\end{equation}
and
\begin{equation}
\Delta_{n+1} = \Delta(n+1,\pi_{\mathsf{V}}(\Delta_{n})\cup\{\varphi(n+1)\})	
\end{equation}	
for each positive integer $n$. Observe that $(\Delta_{n})_{n \in \mathbb{N}}$ is an increasing sequence of elements of $\mathscr{S}$ such that 
$\mathbb{Z}\times \vv = \bigcup\limits_{n \in \mathbb{N}}\Delta_{n}$.

\begin{lemma}\label{lemma1}
Let $\Delta = \Delta_{m}$ for some $m \in \mathbb{N}$, and let $\overline{\Delta}$ be an element of $\mathscr{S}$ defined by
\[\overline{\Delta} = \bigcup\limits_{\substack{\mathsf{x} \in \Delta \\ \pi_{\mathbb{Z}}(\mathsf{x}) = -m}}U(\mathsf{x}).\]	
Given a finite volume configuration $\xi$ in $S^{\Delta}$, the measure $\lambda^{\xi}$ on $(\Omega, \mathscr{F}_{\,\overline{\Delta}})$ 
defined by
\begin{equation}
\lambda^{\xi}(B) = \int_{B}\;\prod\limits_{\mathsf{x} = (n,i) \in \Delta} p_{i}(\xi_{\mathsf{x}}|(\xi \omega_{\Delta^{c}})_{n-1}) \,\mu_{\nu}(d\omega)	
\end{equation}	
can be expressed as
\begin{equation}
\lambda^{\xi}(B) = \int_{B}\mathds{1}_{[\xi]}(\omega)\mu_{\nu}(d\omega).	
\end{equation}
\end{lemma}
\begin{proof}[Proof of Lemma \ref{lemma1}]
It suffices to show the identity for cylinder sets of the form $[\zeta]$, where each $\zeta$ belongs to $S^{\overline{\Delta}}$. The result follows 
by using the fact that the map
\[\omega \mapsto \prod\limits_{\mathsf{x} = (n,i) \in \Delta} p_{i}(\xi_{\mathsf{x}}|(\xi \omega_{\Delta^{c}})_{n-1})\]
depends only on the values of $\omega$ assumed on $\overline{\Delta}$. 	
\end{proof}

\begin{proof}[Proof of Theorem~\ref{stm}]
Let us fix a set $\Lambda \in \mathscr{S}$ and a finite volume configuration $\sigma$ in $S^{\Lambda}$. Let $\Delta = \Delta_{m}$ for some positive integer 
$m$ such that 
\[\{\mathsf{x} \in \mathbb{Z}\times \vv : (\{\mathsf{x}\}\cup U(\mathsf{x})) \cap \Lambda \neq \emptyset\} \subseteq \Delta_{m}.\]
Then, for each $\omega$ in $\Omega$, we have
\begin{eqnarray*}
e^{-H_{\Lambda}^{\Phi}	(\sigma \omega_{\Lambda^{c}})} &=& \prod\limits_{\substack{\mathsf{x} = (n,i) \\ (\{\mathsf{x}\}\cup U(\mathsf{x})) \cap \Lambda \neq \emptyset }} p_{i}((\sigma \omega_{\Lambda^{c}})_{\mathsf{x}}|(\sigma \omega_{\Lambda^{c}})_{n-1})\\
&=& \frac{\prod\limits_{\mathsf{x} = (n,i) \in \Delta} p_{i}((\sigma \omega_{\Lambda^{c}})_{\mathsf{x}}|(\sigma \omega_{\Lambda^{c}})_{n-1})}{\prod\limits_{\substack{\mathsf{x} = (n,i) \in \Delta \\ (\{\mathsf{x}\}\cup U(\mathsf{x})) \cap \Lambda = \emptyset }} p_{i}(\omega_{\mathsf{x}}|\omega_{n-1})},
\end{eqnarray*}
thus
\begin{equation}\label{vtnc}
\frac{e^{-H_{\Lambda}^{\Phi}(\sigma \omega_{\Lambda^{c}})}}{\sum\limits_{\sigma ' \in S^{\Lambda}} e^{-H_{\Lambda}^{\Phi}(\sigma ' \omega_{\Lambda^{c}})}}	
= \frac{\prod\limits_{\mathsf{x} = (n,i) \in \Delta} p_{i}((\sigma \omega_{\Lambda^{c}})_{\mathsf{x}}|(\sigma \omega_{\Lambda^{c}})_{n-1})}{\sum\limits_{\sigma ' \in S^{\Lambda}}\prod\limits_{\mathsf{x} = (n,i) \in \Delta} p_{i}((\sigma' \omega_{\Lambda^{c}})_{\mathsf{x}}|(\sigma' \omega_{\Lambda^{c}})_{n-1})}.
\end{equation}
Now, given a finite volume configuration $\eta$ in $S^{\Delta \backslash  \Lambda}$, using equation (\ref{vtnc}), we obtain
\begin{eqnarray*}
\int_{[\eta]}\mathds{1}_{[\sigma]}(\omega)\mu_{\nu}(d\omega) &=& \lambda^{\sigma\eta}(\Omega) = \int\prod\limits_{\mathsf{x} = (n,i) \in \Delta} p_{i}((\sigma\eta)_{\mathsf{x}}|(\sigma \eta \omega_{\Delta^{c}})_{n-1}) \,\mu_{\nu}(d\omega)\\
&=& \sum\limits_{\zeta \in S^{\Lambda}} \int \frac{e^{-H_{\Lambda}^{\Phi}(\sigma \eta\omega_{\Delta^{c}})}}{\sum\limits_{\sigma ' \in S^{\Lambda}} e^{-H_{\Lambda}^{\Phi}(\sigma ' \eta\omega_{\Delta^{c}})}} \prod\limits_{\mathsf{x} = (n,i) \in \Delta} p_{i}((\zeta\eta)_{\mathsf{x}}|(\zeta \eta \omega_{\Delta^{c}})_{n-1}) \,\mu_{\nu}(d\omega)\\
&=&	\sum\limits_{\zeta \in S^{\Lambda}} \int \frac{e^{-H_{\Lambda}^{\Phi}(\sigma \eta\omega_{\Delta^{c}})}}{\sum\limits_{\sigma ' \in S^{\Lambda}} e^{-H_{\Lambda}^{\Phi}(\sigma ' \eta\omega_{\Delta^{c}})}} \lambda^{\zeta \eta}(d\omega) \\
&=&	\sum\limits_{\zeta \in S^{\Lambda}} \int \frac{e^{-H_{\Lambda}^{\Phi}(\sigma \eta\omega_{\Delta^{c}})}}{\sum\limits_{\sigma ' \in S^{\Lambda}} e^{-H_{\Lambda}^{\Phi}(\sigma ' \eta\omega_{\Delta^{c}})}} \mathds{1}_{[\zeta \eta]} \mu_{\nu}(d\omega) \\
&=& \int_{[\eta]} \frac{e^{-H_{\Lambda}^{\Phi}(\sigma \omega_{\Lambda^{c}})}}{\sum\limits_{\sigma ' \in S^{\Lambda}} e^{-H_{\Lambda}^{\Phi}(\sigma ' \omega_{\Lambda^{c}})}} \mu_{\nu}(d\omega).
\end{eqnarray*}	
Since $(\Delta_{n})_{n \in \mathbb{N}}$ is an increasing sequence of elements of $\mathscr{S}$ such that 
$\mathbb{Z}\times \vv= \bigcup\limits_{n \in \mathbb{N}}\Delta_{n}$, it follows that the equality
\begin{equation}
\mu_{\nu}([\sigma]|\mathscr{F}_{\Lambda^{c}})(\omega) = \frac{e^{-H_{\Lambda}^{\Phi}(\sigma \omega_{\Lambda^{c}})}}{\sum\limits_{\sigma ' \in S^{\Lambda}} e^{-H_{\Lambda}^{\Phi}(\sigma ' \omega_{\Lambda^{c}})}} 	
\end{equation}
holds for $\mu_{\nu}$-almost every point $\omega$ in $\Omega$.
\end{proof}

\subsection{Proof of Theorem \ref{drift}}
Let $m$ and $N$ be integers, where $N \geq 0$, and let us consider the set 
\begin{equation}\label{d}
\Delta = \{m\} \times \{j \in \mathsf{V}: d(\mathbf{o},j) \leq N\}.	
\end{equation}	
If we consider the nonempty finite subset $\Lambda$ of $\mathbb{Z}\times \mathsf{V}$ given by
\begin{equation}\label{l}
\Lambda = \bigcup\limits_{l = 0}^{N}\{m+l\} \times \{j \in \mathsf{V}: d(\mathbf{o},j) \leq N-l\},	
\end{equation}
it follows that
\begin{eqnarray*}
e^{-H_{\Lambda}^{\Phi}(\xi \omega_{\Lambda^{c}})}
&=& \prod\limits_{\mathsf{x} = (n,i) \in \Lambda} p_{i}(\xi_{\mathsf{x}}|(\xi \omega_{\Lambda^{c}})_{n-1}) \cdot \prod\limits_{\substack{\mathsf{x} = (n,i) \notin \Lambda \\ U(\mathsf{x}) \cap \Lambda \neq \emptyset}} p_{i}(\omega_{\mathsf{x}}|(\xi \omega_{\Lambda^{c}})_{n-1})\\
&=& \prod\limits_{\mathsf{x} = (n,i) \in \Lambda} p_{i}(\xi_{\mathsf{x}}|(\xi \omega_{\Lambda^{c}})_{n-1}) \cdot \prod\limits_{\substack{\mathsf{x} = (n,i) \notin \Lambda \\ U(\mathsf{x}) \cap \Lambda \neq \emptyset}} p_{i}(\omega_{\mathsf{x}}|\omega_{n-1})
\end{eqnarray*}
holds for all finite volume configuration $\xi$ in $S^{\Lambda}$ and for every $\omega$ in $\Omega$. Since $\mu$ is a Gibbs measure, then for $\mu$-almost every point $\omega$ in $\Omega$ we have
\begin{equation*}
\mu([\xi]|\mathscr{F}_{\Lambda^{c}})(\omega) = \frac{\prod\limits_{\mathsf{x} = (n,i) \in \Lambda}p_{i}(\xi_{\mathsf{x}}|(\xi \omega_{\Lambda^{c}})_{n-1})}{\sum\limits_{\eta \in S^{\Lambda}}\prod\limits_{\mathsf{x} = (n,i) \in \Lambda}p_{i}(\eta_{\mathsf{x}}|(\eta \omega_{\Lambda^{c}})_{n-1})} 
= \prod\limits_{\mathsf{x} = (n,i) \in \Lambda} p_{i}(\xi_{\mathsf{x}}|(\xi \omega_{\Lambda^{c}})_{n-1}),
\end{equation*}
and summing over all possible spins inside the volume $\Lambda\backslash \Delta$, we conclude that 
\begin{equation}\label{cami}
\mu([\xi_{\Delta}]|\mathscr{F}_{\Lambda^{c}})(\omega) = \prod\limits_{\mathsf{x} = (m,i) \in \Delta} p_{i}(\xi_{\mathsf{x}}|\omega_{m-1}).
\end{equation}
If we define the $\sigma$-algebra $\mathscr{F}_{<m}$ as the $\sigma$-algebra $\mathscr{F}_{\Gamma(m)}$ of subsets of $\Omega$, where 
$\Gamma(m) = \{\mathsf{x} \in S : \pi_{\mathbb{Z}}(\mathsf{x})< m\}$, it follows from (\ref{cami}) that
\begin{equation}\label{lerolero}
\mu(\{\omega' \in \Omega: \omega'_{m} \in B\}|\mathscr{F}_{<m})(\omega) = P(B|\omega_{m-1}) 
\end{equation}
holds for $\mu$-almost every $\omega$ in $\Omega$ and for any measurable subset $B$ of $\Omega_{0}$.

Since $\mu$ is invariant under time translations, it follows that the measure $\nu$ on $(\Omega_{0},\mathscr{B}(\Omega_{0}))$ defined by
\begin{equation}
\nu(B) = \mu(\{\omega' \in \Omega: \omega'_{m} \in B\})
\end{equation}
does not depends on the choice of the integer $m$, moreover, it is easy to show that $\nu$ is a stationary measure for the PCA. Using equation (\ref{lerolero})
and Kolmogorov consistency theorem, we finally conclude that $\mu = \mu_{\nu}$.

\section{Appendix}\label{app2}
\setcounter{equation}{0}
\renewcommand{\theequation}{B\thechapter.\arabic{equation}}
\begin{proof}[Proof of Lemma~\ref{lem:statsol}]
Let us proof that given a function $a: \{-1,+1\}^d \rightarrow \mathbb{R}$ and a probability measure $\mu$ on $\{-1,+1\}^{\mathbb{T}^d}$, we have
\begin{eqnarray}\label{ixala}
\sum_{\xi \in \{-1,+1\}^d}&&a(\xi) \prod_{k \in \{0,\dots,d-1\}}\mu(x_{e_k}= \xi_k) \\
&&=\sum_{l=0}^d (-1)^l \left[\sum_{\substack{I \subseteq \{0, \dots, d-1\} \\ |I| = l}}\left(\sum_{\substack{\xi \in \{-1,+1\}^d \\ \xi_m = -1\;\text{for all $m \in I$}}}(-1)^{\#\{m: \xi_m = -1\}} a(\xi)\right) \prod_{k \in \{0,\dots,d-1\}\backslash I} \mu(x_{e_k}= +1)\right] \nonumber
\end{eqnarray}
We prove the equation above by induction. For the case where $d=1$, we proof is straightforward. If we suppose that the result is proven for $d$, then
\begin{eqnarray*}
\sum_{\xi \in \{-1,+1\}^{d+1}}&&a(\xi) \prod_{k \in \{0,\dots,d\}}\mu(x_{e_k}= \xi_k) \\
&=& \sum_{\xi \in \{-1,+1\}^{d}} a(\xi,+1) \prod_{k \in \{0,\dots,d-1\}}\mu(x_{e_k}= \xi_k) \cdot \mu(x_{e_d}= +1)\\
&&+ \sum_{\xi \in \{-1,+1\}^{d}} a(\xi,-1) \prod_{k \in \{0,\dots,d-1\}}\mu(x_{e_k}= \xi_k) \cdot \mu(x_{e_d}= -1) 
\end{eqnarray*}
\begin{eqnarray*}
&=& \sum_{\xi \in \{-1,+1\}^{d}} a(\xi,+1) \prod_{k \in \{0,\dots,d-1\}}\mu(x_{e_k}= \xi_k) \cdot \mu(x_{e_d}= +1)\\
&&- \sum_{\xi \in \{-1,+1\}^{d}} a(\xi,-1) \prod_{k \in \{0,\dots,d-1\}}\mu(x_{e_k}= \xi_k) \cdot \mu(x_{e_d}= +1) \\
&&+\sum_{\xi \in \{-1,+1\}^{d}} a(\xi,-1) \prod_{k \in \{0,\dots,d-1\}}\mu(x_{e_k}= \xi_k) \\ 
&=& \sum_{l=0}^d (-1)^l \left[\sum_{\substack{I \subseteq \{0, \dots, d-1\} \\ |I| = l}}\left(\sum_{\substack{\xi \in \{-1,+1\}^d \\ \xi_m = -1\;\text{for all $m \in I$}}}(-1)^{\#\{m: \xi_m = -1\}} a(\xi,+1)\right) \prod_{k \in \{0,\dots,d\}\backslash I} \mu(x_{e_k}= +1)\right] \\
&&- \sum_{l=0}^d (-1)^l \left[\sum_{\substack{I \subseteq \{0, \dots, d-1\} \\ |I| = l}}\left(\sum_{\substack{\xi \in \{-1,+1\}^d \\ \xi_m = -1\;\text{for all $m \in I$}}}(-1)^{\#\{m: \xi_m = -1\}} a(\xi,-1)\right) \prod_{k \in \{0,\dots,d\}\backslash I} \mu(x_{e_k}= +1)\right] \\
&&+\sum_{l=0}^d (-1)^l \left[\sum_{\substack{I \subseteq \{0, \dots, d-1\} \\ |I| = l}}\left(\sum_{\substack{\xi \in \{-1,+1\}^d \\ \xi_m = -1\;\text{for all $m \in I$}}}(-1)^{\#\{m: \xi_m = -1\}} a(\xi,-1)\right) \prod_{k \in \{0,\dots,d-1\}\backslash I} \mu(x_{e_k}= +1)\right] \\
&=& \sum_{l=0}^d (-1)^l \left[\sum_{\substack{I \subseteq \{0, \dots, d\} \\ |I| = l, d \notin I}}\left(\sum_{\substack{\xi \in \{-1,+1\}^{d} \\ \xi_m = -1\;\text{for all $m \in I$}}}(-1)^{\#\{m: \xi_m = -1\}} a(\xi,+1)\right) \prod_{k \in \{0,\dots,d\}\backslash I} \mu(x_{e_k}= +1)\right] \\
&&- \sum_{l=0}^d (-1)^l \left[\sum_{\substack{I \subseteq \{0, \dots, d\} \\ |I| = l, d \notin I}}\left(\sum_{\substack{\xi \in \{-1,+1\}^d \\ \xi_m = -1\;\text{for all $m \in I$}}}(-1)^{\#\{m: \xi_m = -1\}} a(\xi,-1)\right) \prod_{k \in \{0,\dots,d\}\backslash I} \mu(x_{e_k}= +1)\right] \\
&&\sum_{l=0}^d (-1)^{l+1} \left[\sum_{\substack{I \subseteq \{0, \dots, d\} \\ |I| = l+1, d \in I}}\left(\sum_{\substack{\xi \in \{-1,+1\}^{d+1} \\ \xi_m = -1\;\text{for all $m \in I$}}}(-1)^{\#\{m: \xi_m = -1\}} a(\xi)\right) \prod_{k \in \{0,\dots,d\}\backslash I} \mu(x_{e_k}= +1)\right] \\
&=& \sum_{l=0}^d (-1)^l \left[\sum_{\substack{I \subseteq \{0, \dots, d\} \\ |I| = l, d \notin I}}\left(\sum_{\substack{\xi \in \{-1,+1\}^{d+1} \\ \xi_m = -1\;\text{for all $m \in I$}}}(-1)^{\#\{m: \xi_m = -1\}} a(\xi)\right) \prod_{k \in \{0,\dots,d\}\backslash I} \mu(x_{e_k}= +1)\right] \\
&&\sum_{l=0}^d (-1)^{l+1} \left[\sum_{\substack{I \subseteq \{0, \dots, d\} \\ |I| = l+1, d \in I}}\left(\sum_{\substack{\xi \in \{-1,+1\}^{d+1} \\ \xi_m = -1\;\text{for all $m \in I$}}}(-1)^{\#\{m: \xi_m = -1\}} a(\xi)\right) \prod_{k \in \{0,\dots,d\}\backslash I} \mu(x_{e_k}= +1)\right] \\
&=& \sum_{l=0}^{d+1} (-1)^l \left[\sum_{\substack{I \subseteq \{0, \dots, d\} \\ |I| = l}}\left(\sum_{\substack{\xi \in \{-1,+1\}^{d+1} \\ \xi_m = -1\;\text{for all $m \in I$}}}(-1)^{\#\{m: \xi_m = -1\}} a(\xi)\right) \prod_{k \in \{0,\dots,d\}\backslash I} \mu(x_{e_k}= +1)\right].
\end{eqnarray*}
Therefore the result follows.

If we consider the the particular case where $a(\xi) = p_{\mathbf{o}}(+1| \xi)$ and $\mu = \mathsf{Bern(p)}^{\otimes \mathbb{T}^d}$ that satisfies (\ref{vemmonstro}), then equation (\ref{bodybuilder}) follows. Now, if we let $a(\xi) = p_{\mathbf{o}}(+1| \xi)$ and 
$\mu = P^n(x,\cdot)$, then equations (\ref{prodoido}) and (\ref{ixala}) implies equation (\ref{bodybuilder2}). 
\end{proof}

\end{document}